%% file: Coord_Shar_Spec_Precod_Distr_CSIT_single.tex
\documentclass[12pt,draftclsnofoot,onecolumn]{IEEEtran}
\IEEEoverridecommandlockouts

\usepackage{}
\usepackage{amssymb}
\usepackage{amsmath}

\usepackage{graphicx,bm,cite,epsfig,amssymb,amsmath,multirow,enumerate,amsthm}

\usepackage{paralist}
\usepackage[T1]{fontenc}
\usepackage[utf8]{inputenc}
\usepackage{authblk}
\usepackage{steinmetz}

\usepackage{psfrag}

\newtheorem{proposition}{Proposition}
\newtheorem{lemma}{Lemma}

\usepackage{algorithm, algorithmic}

\input{Definitions}
\makeatletter
\renewcommand{\maketag@@@}[1]{\hbox{\m@th\normalsize\normalfont#1}}%
\makeatother

\begin{document}

\title{Coordinated Shared Spectrum Precoding with Distributed CSIT}

\author{Miltiades~C.~Filippou,~\IEEEmembership{Member,~IEEE,}
        Paul~de~Kerret,~\IEEEmembership{Member,~IEEE,}
        David~Gesbert,~\IEEEmembership{Fellow,~IEEE,}
			  Tharmalingam~Ratnarajah,~\IEEEmembership{Senior~Member,~IEEE,}
				Adriano~Pastore,
				and George~A.~Ropokis
\thanks{This work is supported by the Seventh Framework Programme for Research of the European Commission under grant number ADEL-619647.}
\thanks{M.~C.~Filippou and T.~Ratnarajah are with the Institute for Digital Communications, University of Edinburgh, Edinburgh EH9 3FG, U.K. (e-mail: \{m.filippou, t.ratnarajah\}@ed.ac.uk).}
\thanks{P.~de~Kerret is with T\'el\'ecom Bretagne, IMT, UMR CNRS 3192 Lab-STICC, France (e-mail: paul.dekerret@telecom-bretagne.eu).}
\thanks{D. Gesbert is with EURECOM, Campus SophiaTech, 450 Route des Chappes, 06410 Biot, France (e-mail: gesbert@eurecom.fr).}
\thanks{A.~Pastore is with Laboratory for Information in Networked Systems, École polytechnique fédérale de Lausanne, Route Cantonale, 1015 Lausanne, Switzerland (e-mail: adriano.pastore@epfl.ch).}
\thanks{G.~A.~Ropokis is with the Computer Technology Institute and Press ``Diophantus'', Rio-Patras 26500, Greece (e-mail: ropokis@noa.gr).}
}

\IEEEpeerreviewmaketitle

\maketitle

\IEEEpeerreviewmaketitle

\begin{abstract}
In this paper, the operation of a Licensed Shared Access (LSA) system is investigated, considering downlink communication. The system comprises of a Multiple-Input-Single-Output (MISO) incumbent transmitter (TX) - receiver (RX) pair, which offers a spectrum sharing opportunity to a MISO licensee TX-RX pair. Our main contribution is the design of a \emph{coordinated} transmission scheme, inspired by the \emph{underlay} Cognitive Radio (CR) approach, with the aim of maximizing the average rate of the licensee, subject to an average rate constraint for the incumbent. In contrast to most prior works on underlay CR, the coordination of the two TXs takes place under a realistic Channel State Information (CSI) scenario, where each TX has sole access to the instantaneous direct channel of its served terminal. Such a CSI knowledge setting brings about a formulation based on the \emph{theory of Team Decisions}, whereby the TXs aim at optimizing a common objective given the same constraint set, on the basis of individual channel information. Consequently, a novel set of applicable precoding schemes is proposed. Relying on statistical coordination criteria, the two TXs cooperate in the lack of any instantaneous CSI exchange. We verify by simulations that our novel coordinated precoding scheme outperforms the standard underlay CR approach.
\end{abstract}

\IEEEpeerreviewmaketitle

\section{Introduction}
\label{sec:intro}
The utilization of the radio spectrum is internationally regulated by governments, with the aim of providing wireless communication services that can be efficiently protected from harmful interference.~Nevertheless, the tremendous spread of wireless services has given rise to a great need for bandwidth, which cannot be satisfied by an exclusivity of spectral allocation.~On the other hand, \emph{spectrum scarcity}, i.e., the phenomenon where the radio spectrum is becoming vastly underutilized, has been discussed in reports, such as the one published by the Federal Communications Commission (FCC) in 2002 \cite{federal2002spectrum}.~The spotted under-utilization has given rise, in its turn, to the notion of \emph{Cognitive Radio} (CR), which has been  suggested as a promising technology in view of increasing wireless spectral efficiency by exploiting the existing \emph{spectrum holes} in time, frequency or space \cite{Mitola1999, Haykin2005}.

Focusing on the \emph{underlay} CR approach, a primary network allows the simultaneous use of its spectral resources by a new-coming (unlicensed) secondary network, given the condition that the latter will utilize the available resources in a way that the interference created by a secondary transmitter (TX) towards a primary receiver (RX) is below a threshold predefined by the primary network \cite{Goldsmith2009, Biglieri2012}.~Under such a setup, efficient schemes have been proposed, emphasizing on multiple-antenna settings, with the aim of maximizing the information rate of the secondary system, subject to given constraints over the harmful interference suffered by primary terminals \cite{Zhang2008, Kim2011, Zhang2009, Huang2010, Le2008}.~Importantly, in practice, the ability of the secondary TX to acquire global, multi-user Channel State Information (CSI) is unrealistic.~At this point, one would suggest an exchange protocol of channel information between TXs.~However, such a protocol would introduce exchange imperfections of channel estimates as well as delays, hence, leading to degraded performance of the designed precoding solution.
Consequently, an extensive literature has emphasized on designing transmission schemes, that are robust to imperfect CSI or merely requiring local channel knowledge [See \cite{Wang2013} and references therein].~In addition, iterative schemes, based on game theory, have been also investigated as a means of avoiding the need for global multi-user CSI exchange, with respect to spectrum sharing scenarios \cite{Wang2010,Scutari2010, Zhong2011}.

Yet, standard spectrum sharing approaches for underlay CR systems, merely focus on designing a transmission scheme for the secondary TX, such that the interference received by a primary terminal would not overcome a certain threshold, beyond which a secure primary connection cannot be established.~However, standardization bodies have lately focused on the design of Authorized or Licensed Shared Access systems (termed as ASA and LSA) \cite{cept2014report}, \cite{holdren2012realizing}.~The key difference between the latter systems and underlay CR systems is that the incumbents (equivalent to primary nodes in a CR system) can share the spectrum with the licensees (the licensed equivalent of secondary nodes in a CR system), \emph{provided} that Quality-of-Service (QoS) metrics, that have been negotiated prior to licensing, are satisfied for all involved entities.~Motivated by this new framework, it is evident that a major drawback of standard, interference temperature-based underlay CR systems consists in the lack of coordination between the primary and the secondary systems.~As a result, in such a way, the primary system tends to overspend its available resources, leading to poor throughput performance at the secondary side.

Given this situation, in this work, we propose the design of a coordination scheme for the two TXs, based on commonly available, slow-varying statistical information.~The goal of such a design will be the maximization of the (average) information rate of the licensee system, given that the achievable (average) rate of the incumbent system lies above a certain threshold.
More precisely, according to this scheme, the incumbent TX exploits its locally available CSI and the statistical (covariance) information of the global multi-user channel to coordinate with the licensee TX, in order to guarantee the desired average throughput for its assigned terminal.~Each TX has access to the instantaneous direct links of its assigned users (as well as their statistics), whereas, the interference cross-links are just statistically known at both TXs.~The availability of different estimates of the global downlink channel falls within the paradigm of \emph{Team Decision} theory, because both transmitters (incumbent and licensee) are actively engaged in cooperation, while being constrained by the locality of the available instantaneous CSI \cite{Radner1962, Ho1980, Zakhour2010, DeKerret2013, Filippou2013}.

In \cite{DeKerret2014}, a similar scenario was investigated for a single-user underlay CR setup, in the existence of spatially uncorrelated direct channel links, while, in the present work, we focus on the performance of an extended set of applicable joint precoding schemes, with the assumption of correlated Rayleigh fading for all the involved MISO channels.~More particularly, our contributions are the following:

\begin{itemize}
 \item We design a low complexity, statistically coordinated precoding scheme for a MISO spectrum sharing system, which can be applicable to a shared spectrum access system (ASA or LSA).~The goal of this design is the maximization of the throughput of the licensee network, in terms of the achievable average rate, given a QoS constraint on the average rate of the incumbent.

\item In order to design this joint precoding scheme, we derive lower bounds for the average rate of each RX, in closed form, both when a precoding solution based on Matched Filtering (MF) and statistical Zero-Forcing (sZF), is applied.~This way, since the original optimization problem is hard to solve, because of the existence of combined instantaneous and statistical CSI at each of the two TXs, we focus on an approximated version of it.

\item Focusing on the approximated version of the optimization problem, we design a set of applicable, joint precoding solutions, the elements of which are joint transmission schemes based on MF or sZF-based beamforming solutions.

\item For each of the joint transmission schemes that belong to the referred set, a power policy coordination criterion, which is based on the statistics of the global downlink channel, is applied, with the aim of finding the average transmit power levels at each TX.~These power levels are such that the average rate criterion on the incumbent RX is satisfied.~Then, the joint beamforming and power allocation scheme that maximizes the average throughput of the licensee RX is selected for transmission.~Since such a decision relies on commonly available statistical information, it is taken by both TXs coherently.

\item The novel, coordinated precoding scheme is numerically evaluated in comparison to the standard interference temperature-based underlay CR approach.~It is shown that our scheme outperforms the standard underlay CR one in a range of system scenarios, which makes our designed precoding policy appealing for LSA systems, since higher average throughput is achieved for the (candidate) licensee system.
\end{itemize}

Throughout the paper, the following notations are adopted: all boldface letters indicate vectors (lower case) or matrices (upper case).~${\mathbf{A}}^{\He}$, $\tr(\mathbf{A})$ and ${[\mathbf{A}]}_{m,n}$ denote the Hermitian transpose of matrix $\mathbf{A}$, its trace, and its $(m,n)$-th entry, respectively, whereas $\lambda_j(\mathbf{A})$ stands for its $j$-th eigenvalue.~Also, $\diag\left(\alpha_1,\ldots,\alpha_n\right)$ symbolizes a diagonal matrix, the elements of which are $\alpha_1,\ldots,\alpha_n$.~Additionally, $\mathbb{E}[\cdot]$ symbolizes the expectation operator and $\|\cdot\|$ denotes the Euclidean norm, while $\bm{0}_n$ denotes the all-zero vector of dimension $n$.~The identity matrix of dimension $n \times n$ is denoted by $\mathbf{I}_n$, while $\bar{i}$ denotes the complementary index of $i$, when the cardinality of the considered set is equal to two, i.e., $\bar{i} = i$ mod $2 + 1$.~The phase angle between two vectors $\bm{a}$ and $\bm{b}$ is denoted as $\phase{\bm{a}, \bm{b}}$ and it is defined as: $\phase{\bm{a}, \bm{b}} = \cos^{-1}\left(\frac{\bm{a}^H \bm{b}}{\|\bm{a}\| \|\bm{b}\|}\right)$, where $\cos^{-1}(\cdot)$ stands for the inverse cosine function.~For a random vector $\bm{x}$, $\bm{x} \sim \mathcal{CN}(\boldsymbol{\mu}, \mathbf{\Sigma})$ denotes that $\bm{x}$ follows a Circularly Symmetric Complex Gaussian (CSCG) distribution with mean $\boldsymbol{\mu}$ and covariance matrix $\mathbf{\Sigma}$.~Finally, $E_1(\cdot)$ represents the exponential integral function, which is defined in \cite[eq. (5.1.1)]{Abramowitz}, while $\gamma \approx 0.5772$ stands for the Euler-Mascheroni constant, as it is defined in \cite[eq. (4.1.32)]{Abramowitz}.

\section{System and Channel Model}
%
The spectrum sharing system, which is illustrated in Fig.~\ref{fig:SystemModel}, is composed of a MISO incumbent system, comprising of a TX, ${\textrm{TX}}~1$, equipped with $M_1$ antennas, along with its assigned single-antenna terminal, ${\textrm{RX}}~1$.~Focusing on downlink communication, the incumbent system is willing to share its resources with a MISO licensee system.~The latter system consists of a multiple antenna TX, ${\textrm{TX}}~2$, equipped with $M_2$ antennas, as well as of a licensee terminal, ${\textrm{RX}}~2$, assigned to ${\textrm{TX}}~2$.

Considering the involved channels, spatially correlated Rayleigh fading is assumed for both direct and interfering channel links.~As a consequence, for the channel between ${\textrm{TX}}~j$ and ${\textrm{RX}}~i$, we have: $\bm{h}_{i,j} \sim \mathcal{CN}(\bm{0}_{M_j},  \mathbf{R}_{i,j})$.
\psfrag{h11}{$\bm{h}_{1,1}$}
\psfrag{h22}{$\bm{h}_{2,2}$}
\psfrag{h12}{$\bm{h}_{1,2}$}
\psfrag{h21}{$\bm{h}_{2,1}$}
\psfrag{R11}{$\mathbf{R}_{1,1}$}
\psfrag{R12}{$\mathbf{R}_{1,2}$}
\psfrag{R21}{$\mathbf{R}_{2,1}$}
\psfrag{R22}{$\mathbf{R}_{2,2}$}
\psfrag{TX1}{${\textrm{TX}}~1$}
\psfrag{TX2}{${\textrm{TX}}~2$}
\psfrag{RX1}{${\textrm{RX}}~1$}
\psfrag{RX2}{${\textrm{RX}}~2$}
\psfrag{w1}{$\bm{w}_1$}
\psfrag{w2}{$\bm{w}_2$}
\psfrag{CSIT_at_TX_1}{CSIT at ${\textrm{TX}}~1$}
\psfrag{CSIT_at_TX_2}{CSIT at ${\textrm{TX}}~2$}

\begin{figure}[!ht]
  \centering
  \includegraphics[scale=0.65]{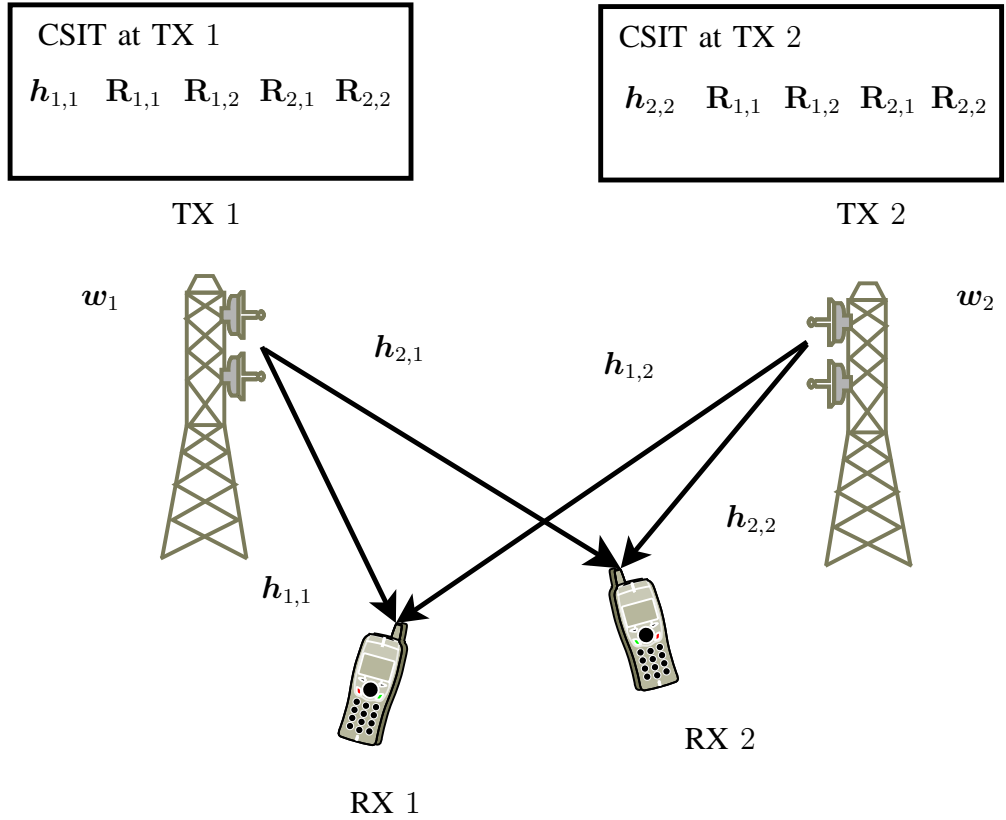}\\
	\vspace{-0.1in}
  \caption{The examined LSA system (post-licensing phase).} \label{fig:SystemModel}
\end{figure}
The signal received at ${\textrm{RX}}~i, \hspace{0.1in} i \in \{1,2\}$, can be expressed as
\begin{equation}
  y_i = \bm{h}_{i,i}^{\He} \bm{w}_i s_i + \bm{h}_{i,\bar{i}}^{\He} \bm{w}_{\bar{i}} s_{\bar{i}} + n_i,
\label{eq:signal_RXi}
\end{equation}
where, $\bm{w}_i$ denotes the transmit beamforming vector at ${\textrm{TX}}~i$ and it is assumed that $\bm{w}_i = \sqrt{P_i} \bm{u}_i$, with $P_i \leq P_i^{\max}$ and $\|\bm{u}_i\| = 1$, where $P_i^{\max}$ is a maximum instantaneous power level at ${\textrm{TX}}~i$.~Also, Gaussian noise is considered at ${\textrm{RX}}~i$, i.e., $n_i \sim \mathcal{CN}(0,N_0)$ and we assume that the information symbols for transmission are taken from a standard complex Gaussian codebook, i.e., $s_i \sim \mathcal{CN}(0,1), \hspace{0.1in} i \in \{1,2\}$.~By analyzing \eqref{eq:signal_RXi}, the instantaneous information rate of ${\textrm{RX}}~i, \hspace{0.1in} i \in \{1,2\}$ is given by \cite{cover2012elements}
\begin{equation}
  R_i = \log_2\Bigg(1 + \frac{P_i {|\bm{h}_{i,i}^{\He} \bm{u}_i|}^2}{N_0 +  P_{\bar{i}} {|\bm{h}_{i,\bar{i}}^{\He} \bm{u}_{\bar{i}}|}^2}\Bigg).
\label{eq:inst_rate_RXp}
\end{equation}
In the section that follows, the problem of joint downlink precoding with combined, local CSI at the TX (CSIT), is formulated.

\section{Problem Formulation}

Focusing on the described system model, a realistic CSIT assumption that can be made is that ${\textrm{TX}}~i, \hspace{0.1in} i\in\{1,2\}$, has both instantaneous and statistical (covariance) knowledge of its direct links (i.e., ${\textrm{TX}}~1$ has instantaneous knowledge of direct link $\bm{h}_{1,1}$ and ${\textrm{TX}}~2$ has instantaneous knowledge of direct link $\bm{h}_{2,2}$), whereas, the interference cross-links are merely statistically known via knowledge of their covariance matrices. The second order statistics of the involved channels constitute slow-varying information that can be realistically collected by each TX through low capacity/high delay links.

Capitalizing on the available CSIT at ${\textrm{TX}}~i, \hspace{0.1in} i \in \{1,2\}$, the optimization problem of maximizing the average rate of the licensee system, subject to an average rate constraint for ${\textrm{RX}}~1$ can be formulated as a \emph{functional} optimization problem, with functional dependencies related to the available CSI.~Hence, the resulting optimization problem can be described as follows
\begin{equation}
\begin{aligned}
  &\big(\bm{w}_1^{*}, \bm{w}_2^{*}\big) = \arg \max \E\LSB R_2\left(\bm{w}_1(\bm{h}_{1,1}), \bm{w}_2(\bm{h}_{2,2})\right)\RSB \\
	& {\textrm{subject}} \hspace{0.1in} {\textrm{to}} \hspace{0.1in} \E\LSB R_1\left(\bm{w}_1(\bm{h}_{1,1}), \bm{w}_2(\bm{h}_{2,2})\right)\RSB \geq \tau_1 > 0, \\
	& 0 \leq {\|\bm{w}_1(\bm{h}_{1,1})\|}^2 \leq P_1^{\max}, \hspace{0.1in} 0 \leq {\|\bm{w}_2(\bm{h}_{2,2})\|}^2 \leq P_2^{\max},
\end{aligned}
\tag{P1}
\label{eq:prob_orig}
\end{equation}
where $\tau_1$ stands for the QoS demand of $\textrm{RX}~1$, in terms of average rate.

The dependencies of the precoders to be optimized, on the corresponding instantaneous links, can be formulated by the following mappings
\begin{equation}
\begin{aligned}
  \bm{w}_i: \hspace{0.1in} \mathbb{C}^{M_i} &\rightarrow \mathbb{C}^{M_i} \\
	                             \bm{h}_{i,i} &\mapsto \bm{w}_i(\bm{h}_{i,i}),
\end{aligned}
\label{eq:depend_mapping_TXi}
\end{equation}
where $i \in \{1,2\}$.~Hence, the formulated optimization problem is of distributed nature.~In what follows, for ease of exposition, we will omit to mention explicitly the dependencies of the precoders.

One has to ensure that problem \eqref{eq:prob_orig} is feasible, in other words, QoS threshold~$\tau_1$ has to be achievable when ${\textrm{TX}}~1$ transmits with full power MF precoding, in the absence of any interference coming from the licensee. This means satisfying the following
\begin{equation}
  \mathbb{E}\LSB\log_2\bigg(1 + \frac{P_1^{\max} \|\bm{h}_{1,1}\|^2}{N_0}\bigg)\RSB \geq \tau_1.
\label{eq:achievable_tau}
\end{equation}
The expectation over the interfering channels makes the optimization difficult to handle. However, exploiting the convexity of function $\log_2\left(1 + \frac{1}{x}\right)$, it becomes possible to apply Jensen's inequality \cite{cover2012elements} over the interfering channels. This significantly simplifies the optimization problem, while preserving its important features. The average rate expression for RX~$i$, thus, becomes
\begin{equation}
\begin{aligned}
  \E\LSB R_i\RSB &= \E_{\bm{h}_{i,i},\bm{h}_{i,\bar{i}}} \LSB\log_2\left(1 + \frac{P_i {|\bm{h}_{i,i}^{\He} \bm{u}_i|}^2}{N_0 + P_{\bar{i}} {|\bm{h}_{i,\bar{i}}^{\He} \bm{u}_{\bar{i}}|}^2}\right)\RSB \\
	&\geq \E_{\bm{h}_{i,i}} \LSB \log_2\left(1 + \frac{P_i {|\bm{h}_{i,i}^{\He} \bm{u}_i|}^2}{N_0 + \E_{\bm{h}_{i,\bar{i}}}\LSB P_{\bar{i}} {|\bm{h}_{i,\bar{i}}^{\He} \bm{u}_{\bar{i}}|}^2\RSB}\right)\RSB \\
	&= \E_{\bm{h}_{i,i}} \LSB\log_2\left(1 + \frac{P_i {|\bm{h}_{i,i}^{\He} \bm{u}_i|}^2}{N_0 + P_{\bar{i}} \bm{u}_{\bar{i}}^{\He} \mathbf{R}_{i,\bar{i}} \bm{u}_{\bar{i}}}\right)\RSB \\
	&\triangleq \E\LSB\tilde{R}_i(\bm{w}_i, \bm{w}_{\bar{i}})\RSB.
\end{aligned}
\label{eq:jensen_p}
\end{equation}
\begin{remark}
It should be noted that this approach is only possible thanks to the fact that the precoders~$\bm{w}_1$ and~$\bm{w}_2$ are independent of the instantaneous cross-channels (as only the direct links are instantaneously known).\qed
\end{remark}

With the aim of deriving a practical solution, slow power control depending on the long term statistical channel information, is assumed.~Hence, instead of (instantaneous) power levels $P_1$ and $P_2$, we can use slow power allocation levels $\bar{P}_1$ and $\bar{P}_2$, where $0 \leq \bar{P}_i \leq P_i^{\max}, \hspace{0.1in} i \in \{1,2\}$.

Altogether, in the remainder of the paper, we will work on the following optimization problem:
\begin{equation}
\begin{aligned}
 &(\bar{P}_1^{*}, \bm{u}_1^{*}, \bar{P}_2^{*}, \bm{u}_2^{*}) = \arg \max ~~\E\LSB \tilde{R}_2(\bar{P}_1, \bm{u}_1, \bar{P}_2, \bm{u}_2)\RSB \\
&~~~~~~~~~~~~~~~~~~~~~~~~~~~~~~~~~ {\textrm{subject}} \hspace{0.1in} {\textrm{to}} \hspace{0.1in} \E\LSB \tilde{R}_1(\bar{P}_1, \bm{u}_1, \bar{P}_2, \bm{u}_2)\RSB \geq \tau_1, \\
&~~~~~~~~~~~~~~~~~~~~~~~~~~~~~~~~~ 0 \leq \bar{P}_1 \leq P_1^{\max}, \hspace{0.1in} 0 \leq \bar{P}_2 \leq P_2^{\max}, 
\\
&~~~~~~~~~~~~~~~~~~~~~~~~~~~~~~~~~ \|\bm{u}_1\| = 1, \hspace{0.1in} \|\bm{u}_2\| = 1.
\end{aligned}
\tag{P2}
\label{eq:prob_new}
\end{equation}
%

\section{Preliminary Results}\label{se:preliminaries}
%
The following two propositions provide some characteristics of the optimal solution of problem~\eqref{eq:prob_new}.
\begin{proposition}
\label{propos_1}
The ergodic rate constraint of ${\textrm{RX}}~1$ is satisfied with equality by any optimal solution of \eqref{eq:prob_new}, i.e., 
\begin{equation}
 \mathbb{E}\LSB \tilde{R}_1(\bar{P}_1^{*}, \bm{u}_1^{*}, \bar{P}_2^{*}, \bm{u}_2^{*})\RSB = \tau_1.
\label{eq:propos_1_rate}
\end{equation}
\end{proposition}
\begin{proof}
 The objective $\E\LSB \tilde{R}_2(\bar{P}_1, \bm{u}_1, \bar{P}_2, \bm{u}_2)\RSB$ is monotonically decreasing with respect to $\bar{P}_1$, while, on the other hand, the constraint $\E\LSB \tilde{R}_1(\bar{P}_1, \bm{u}_1, \bar{P}_2, \bm{u}_2)\RSB$ is monotonically increasing and continuous in $\bar{P}_1$.~As a result, one can increase the objective by reducing power level $\bar{P}_1$ up to the point, where the average rate constraint of RX~$1$ will be satisfied with equality. This is always feasible because $\tau_1 > 0$ implies that $\bar{P}_1^{*} > 0$.
\end{proof}
The second proposition yields some insight with respect to the optimal power allocation scheme.
\begin{proposition}
\label{propos_2}
An optimal solution of problem~\eqref{eq:prob_new} satisfies that either TX~$1$ or TX~$2$ transmits with full power, i.e., when $\bar{P}_1^{*} = P_1^{\max}$ or $\bar{P}_2^{*} = P_2^{\max}$.
\end{proposition}

\begin{proof}
Considering an optimal solution, one can write $\bar{P}_1^{*} = \alpha_1^{*} \bar{P}$, for some $\alpha_1^{*} \geq 0$ and $\bar{P}_2^{*} = \alpha_2^{*} \bar{P}$, for some $\alpha_2^{*} \geq 0$, where $\bar{P} > 0$.~Then, taking every term of the objective and dividing the numerator and the denominator of its Signal to Interference plus Noise Ratio (SINR) by $\bar{P}$, one obtains
\begin{equation}
 \E\LSB \tilde{R}_2(\bm{w}_1^{*}, \bm{w}_2^{*})\RSB = \E\LSB \log_2\left(1 + \frac{\alpha_2^{*}{|\bm{h}_{2,2}^{\He} \bm{u}_2^{*}|}^2}{\frac{N_0}{\bar{P}} + \alpha_1^{*} {(\bm{u}_1^{*})}^{\He} \mathbf{R}_{2,1} \bm{u}_1^{*}}\right)\RSB,
\label{propos_2_RXs_k_rate}
\end{equation}
which is a monotonically increasing function of $\bar{P}$.~Similarly, the achievable average rate at ${\textrm{RX}}~1$ becomes
\begin{equation}
 \E\LSB \tilde{R}_1(\bm{w}_1^{*}, \bm{w}_2^{*})\RSB= \E\LSB \log_2\left(1 + \frac{\alpha_1^{*}{|\bm{h}_{1,1}^{\He} \bm{u}_1^{*}|}^2}{\frac{N_0}{\bar{P}} + \alpha_2^{*} {(\bm{u}_2^{*})}^{\He} \mathbf{R}_{1,2} \bm{u}_2^{*}}\right)\RSB,
\label{propos_2_RXp_rate}
\end{equation}
which is a monotonically increasing function of $\bar{P}$, as well.

If none of the two TXs transmits with full power, it means that it is possible to transmit with $\bar{P}^{\prime} > \bar{P}$.~Thus, the transmission using $(\alpha_1^{*} \bar{P}^{\prime}, \bm{u}_1^{*}, \alpha_2^{*} \bar{P}^{\prime}, \bm{u}_2^{*})$ is feasible and leads to a larger objective, which contradicts the optimality of $(\alpha_1^{*} \bar{P}, \bm{u}_1^{*}, \alpha_2^{*} \bar{P}, \bm{u}_2^{*})$.
\end{proof}

\section{Statistically Coordinated Precoding}\label{se:main}

We now present our main contribution which is a new transmission scheme constituting a possible solution for optimization problem \eqref{eq:prob_orig}. Indeed, it is important to note that, although possibly suboptimal, our approach is able to \emph{guarantee} the incumbent rate constraint and is, therefore, a solution to the initial optimization problem.

\subsection{General Approach}
%
Since the derivation of closed-form expressions for the optimal precoders is hardly tractable due to the functional nature of optimization problem \eqref{eq:prob_new} (which requires optimizing over an infinite dimensional space), we discretize the functional space and restrict the space of possible precoding solutions to a \emph{set of transmission strategies}, $\mathcal{S} \triangleq \{\mathcal{S}_1,\ldots,\mathcal{S}_L\}$, where $L$ denotes the number of different joint transmission schemes applied by ${\textrm{TX}}~1$ and ${\textrm{TX}}~2$. 

More specifically, each element, $\mathcal{S}_l \in \mathcal{S}$, refers to an ordered set $(\bar{P}_{1,\mathcal{S}_l}, \bm{u}_{1,\mathcal{S}_l}, \bar{P}_{2,\mathcal{S}_l}, \bm{u}_{2,\mathcal{S}_l}), \hspace{0.1in} l \in \{1,\ldots,L\}$. 

Such a restriction to a finite set of joint transmission schemes allows for every transmit strategy, $\mathcal{S}_l, \hspace{0.1in} l=1,\ldots,L$, to be evaluated in terms of feasibility and in terms of performance. Furthermore, it provides a low complexity method for coordinating the TXs.

\subsubsection{Beamforming Design}
The first step consists in designing the strategy set, i.e., the beamforming strategies. Although any beamforming scheme could be chosen in theory, a good heuristic choice is key to the tractability and the efficiency of the approach. In this work, we restrict our analysis to the MF and the sZF strategies, as they represent the extreme approaches between which it will be necessary to strike a trade-off.

MF precoding corresponds to the \emph{egoistic} beamforming scheme, where TX~$i$ transmits using 
\begin{equation}
\bm{u}_{i,{\textrm{MF}}} \triangleq \frac{\bm{h}_{i,i}}{\|\bm{h}_{i,i}\|}.
\end{equation}
This beamformer maximizes the strength of the direct link without any consideration of the interference.

In constrast, sZF corresponds to an \emph{altruistic} beamforming scheme, where TX~$i$ transmits using
\begin{equation}
\bm{u}_{i,{\textrm{sZF}}} = \arg \displaystyle \max_{\bm{u} \in \mathbb{C}^{M_i \times 1}} \bm{u}^{\He} \mathbf{R}_{\bar{i},i}^{-\frac{1}{2}} \mathbf{R}_{i,i} \mathbf{R}_{\bar{i},i}^{-\frac{1}{2}} \bm{u}.
\end{equation}
The sZF beamforming scheme consists in exploiting the statistical information of the cross-links to reduce the created interference, while also taking into consideration the statistical information of the direct links. This strategy has the advantage of using only statistical information available at both TXs and hence enforces perfect \emph{coordination} between the TXs, which will prove critical to an efficient joint transmission scheme.

\subsubsection{Power Control Policy}

Power control is a key ingredient to ensure that the average rate constraint for the incumbent RX is not violated. Furthermore, it is shown in Section~\ref{se:preliminaries} that the incumbent QoS constraint is always fulfilled with equality and that one of the two TXs emits with full power, while the other reduces its power to respect the incumbent constraint. Therefore, we denote by $\mathcal{P}_1$ the joint power policy where TX~$1$ emits with full power and by $\mathcal{P}_2$ the joint power policy where TX~$2$ transmits with full power.

\subsubsection{Choice of the Transmission Policy}
Considering the potential applicability of the two power control policies for each of the joint beamforming solutions, such a formulation leads to a joint transmission strategy set, which consists of $8$~possible transmission schemes. However, the incumbent constraint is only fulfilled for some of the strategies and has to be verified otherwise. It is, hence, necessary to compute for each of these $8$~transmission schemes the power emitted by one of the TXs and then evaluate the ergodic rate of both RXs. Once this is done, the best solution, in terms of average throughput for the licensee RX, is directly obtained.

\begin{remark}
It is critical to understand that both TXs will always agree on which strategy to use as the TXs are \emph{statistically coordinated}: only statistical information is necessary to evaluate the ergodic rates and choose the best strategy.\qed
\end{remark}
\subsection{Computation of the Ergodic Rates for each Strategy}
The ergodic rates for each of the $8$~strategies need to be evaluated. However, the expressions are practically the same in the sense that the $8$ possible strategies come from the combination of only a few parameters. We will hence only present in full detail two strategies: MF-MF-$\mathcal{P}_1$ and sZF-sZF-$\mathcal{P}_2$. The expressions for the other strategies can be trivially deduced.

\begin{remark}
The feasibility of a given strategy has to be verified. However, the feasibility of the optimization problem is preserved as the feasibility is guaranteed for strategy MF-MF-$\mathcal{P}_1$. Indeed, it contains the case where TX~$1$ transmits using MF and full power, while TX~$2$ does not transmit at all.\qed
\end{remark} 

\subsubsection{Strategy MF-MF-$\mathcal{P}_1$}

The TXs transmit using the beamforming vectors~$\bm{u}_{1,{\textrm{MF}}}$ and~$\bm{u}_{2,{\textrm{MF}}}$. Furthermore, TX~$1$ transmits using $\bar{P}_1=P_1^{\max}$. It, thus, remains to determine how TX~$2$ controls its power to ensure that the incumbent ergodic rate constraint is fulfilled, i.e., that 
\begin{equation} 
\E\LSB R_1\RSB \geq \tau_1.
\end{equation}
This can then be rewritten as
\begin{equation}
\begin{aligned}
&\E\LSB R_1\RSB\\
&\geq\mathbb{E}\LSB\log_2\left(1 + \frac{P^{\max}_{1} {|\bm{h}_{1,1}^{\He} \bm{u}_{1,\textrm{MF}}|}^2}{N_0 + \bar{P}_2 \bm{u}_{2,\textrm{MF}}^{\He} \mathbf{R}_{1,2} \bm{u}_{2,\textrm{MF}}}\right)\RSB \\
&= \E_{\bm{h}_{1,1},\bm{h}_{2,2}} \LSB\log_2\left(1 + \frac{P^{\max}_{1} {\|\bm{h}_{1,1}\|}^2}{N_0 + \bar{P}_2 \frac{\bm{h}_{2,2}^{\He} \mathbf{R}_{1,2}\bm{h}_{2,2}}{{\|\bm{h}_{2,2}\|}^2}}\right)\RSB \\
&\stackrel{\text{(a)}}\geq \E_{\bm{h}_{1,1}} \LSB\log_2\left(1 + \frac{P^{\max}_1 \|\bm{h}_{1,1}\|^2}{N_0 + \bar{P}_2 \E_{\bm{h}_{2,2}} \LSB\frac{\bm{h}_{2,2}^{\He} \bR_{1,2}\bm{h}_{2,2}}{{\|\bh_{2,2}\|}^2}\RSB}\right)\RSB \geq \tau_1,
\end{aligned} 
\label{MF_MF_crit_1}
\end{equation}
where $(a)$ holds by applying Jensen's inequality to convex function $\log_2\left(1 + \frac{1}{x}\right)$ and the expectation in the denominator can then be computed using Lemma~\ref{lemma_3} in the Appendix with $\bA=\bR_{2,2}$ and $\bB=\bR_{2,2}^{\frac{1}{2}}\bR_{1,2}\bR_{2,2}^{\frac{1}{2}}$.

Finally, a closed form expression for the ergodic rate is obtained with Lemma~\ref{lemma_1}. Hence, the value of $\bar{P}_2$ can be deduced by bisection, in order for the lower bound derived in \eqref{MF_MF_crit_1} to be equal to $\tau_1$.

It remains to evaluate the corresponding achievable average rate of RX~$2$. Following a similar approach as the one for the ergodic rate of the incumbent, we can obtain the following lower bound:
\begin{equation}
\begin{aligned}
&\E\LSB R_2\RSB\\
&\geq\mathbb{E}\LSB\log_2\left(1 + \frac{\bar{P}_2 {|\bm{h}_{2,2}^{\He} \bm{u}_{2,\textrm{MF}}|}^2}{N_0 + P_1^{\max} \bm{u}_{1,\textrm{MF}}^{\He} \mathbf{R}_{2,1} \bm{u}_{1,\textrm{MF}}}\right)\RSB \\
&= \E_{\bm{h}_{1,1},\bm{h}_{2,2}} \LSB\log_2\left(1 + \frac{\bar{P}_2 {\|\bm{h}_{2,2}\|}^2}{N_0 + P_1^{\max} \frac{\bm{h}_{1,1}^{\He} \mathbf{R}_{2,1}\bm{h}_{1,1}}{{\|\bm{h}_{1,1}\|}^2}}\right)\RSB \\
&\geq \E_{\bm{h}_{2,2}} \LSB\log_2\left(1 + \frac{\bar{P}_2 \|\bm{h}_{2,2}\|^2}{N_0 + P_1^{\max} \E_{\bh_{1,1}} \LSB\frac{\bh_{1,1}^{\He} \bR_{2,1}\bh_{1,1}}{{\|\bh_{1,1}\|}^2}\RSB}\right)\RSB.
\end{aligned} 
\label{MF_MF_2}
\end{equation}
 Once more, the expectation in the denominator is obtained using Lemma~\ref{lemma_3}, while a closed form expression for the ergodic rate is obtained with Lemma~\ref{lemma_1}.

\subsubsection{Strategy sZF-sZF-$\mathcal{P}_2$}
In this strategy, the TXs transmit using the beamformers~$\bm{u}_{1,{\textrm{sZF}}}$ and~$\bm{u}_{2,{\textrm{sZF}}}$, while TX~$2$ transmits using $\bar{P}_2=P_2^{\max}$. It remains then to determine $\bar{P}_1$. In that setting, the rate of $\textrm{RX}~1$ can be lower bounded as
\begin{equation}
\begin{aligned}
\E\LSB \tilde{R}_1\RSB&=\E_{\bm{h}_{1,1}} \LSB\log_2\left(1 + \frac{\bar{P}_1 |\bm{h}_{1,1}^{\He} \bm{u}_{1,\textrm{sZF}}|^2}{N_0 + P_2^{\max} \bm{u}_{2,\textrm{sZF}}^{\He} \mathbf{R}_{1,2} \bm{u}_{2,\textrm{sZF}}}\right)\RSB\geq \tau_1.
\end{aligned}
\label{SZF_SZF_coord_1}
\end{equation}
This rate can be directly computed in closed form using Lemma~\ref{lemma_2}. Finally, the power $\bar{P}_1$, such that the ergodic rate constraint for $\textrm{RX}~1$ is met by the derived lower bound, can be obtained by bisection.

It now remains to evaluate the corresponding ergodic rate of the licensee RX. This is given by the following expression
\begin{equation}
\begin{aligned}
\E\LSB \tilde{R}_2\RSB&=\E_{\bm{h}_{2,2}} \LSB\log_2\left(1 + \frac{P_2^{\max} |\bm{h}_{2,2}^{\He} \bm{u}_{2,\textrm{sZF}}|^2}{N_0 + \bar{P}_1 \bm{u}_{1,\textrm{sZF}}^{\He} \mathbf{R}_{2,1} \bm{u}_{1,\textrm{sZF}}}\right)\RSB.
\end{aligned}
\label{SZF_SZF_coord_2}
\end{equation}
The latter expression can be computed in closed form by applying Lemma \ref{lemma_2}.

\section{Reference Precoding Schemes}
In this section, we present two schemes which will be used to evaluate the efficiency of our statistically coordinated precoding approach. 

The first one, denoted as ``interference temperature-based'' precoding, is an adaptation of the approaches in the literature to allow for a fair comparison. Intuitively, it corresponds to the conventional underlay CR paradigm, where solely the secondary TX adapts its strategy in order for the interference received by the primary RX to be below a given threshold \cite{Biglieri2012}. 

The second one constitutes a coordination benchmark and it is a priori not reachable, but allows to bound the sub-optimality of the proposed approach.

\subsection{Interference Temperature-Based Precoding}
The interference temperature approach, extensively used in the CR literature, consists in forcing the secondary TX to create less interference to the primary user than a given interference threshold, which is here denoted by~$\mathcal{I}$. 

Considering that the secondary TX aims at minimizing the interference created and transmits using $\bm{u}_{2,{\textrm{sZF}}}$, the power emitted by the secondary TX is then given by
\begin{equation}
 \bar{P}_2 = \min\left\{\frac{\mathcal{I}}{\bm{u}_{2,{\textrm{sZF}}}^{\He} \mathbf{R}_{1,2} \bm{u}_{2,{\textrm{sZF}}}}, P^{\max}_2\right\} .
\label{int_temp_sec_power_total}
\end{equation}
In order to conduct a fair comparison with the designed statistically coordinated precoding scheme, we need to determine the interference temperature, $\mathcal{I}$, such that the ergodic rate constraint of $\textrm{RX}~1$ is met with equality, i.e.,
\begin{equation}
  \mathbb{E}\LSB\log_2\bigg(1 + \frac{P^{\max}_1 {\|\bm{h}_{1,1}\|}^2}{N_0 + \mathcal{I}}\bigg)\RSB = \tau_1.
\label{int_temp_level}
\end{equation}
The expectation appearing in \eqref{int_temp_level} can be computed by applying Lemma~\ref{lemma_1}. The interference temperature threshold,~$\mathcal{I}$, can be then easily found by bisection.

\subsection{Coordination Benchmark}
When designing the beamformers, we can observe a clear trade-off between maximizing the desired signal (using MF) and minimizing the interference created. Hence, if we assume that one can achieve both goals at the same time, the following optimization problem is obtained for the power control, and leads to an, a priori, infeasible performance upperbound.
\begin{equation}
\begin{aligned}
 &\max_{\bar{P}_1,\bar{P}_2} \E \LSB\log_2\bigg(1 + \frac{\bar{P}_2 \|\bm{h}_{2,2}\|^2}{N_0 + \bar{P}_1\lambda_{\min}\LB\mathbf{R}_{2,1}\RB }\bigg)\RSB \\
 &{\textrm{subject}} \hspace{0.1in} {\textrm{to}}\;\;\;\;\E \LSB\log_2\bigg(1 + \frac{\bar{P}_1 \|\bm{h}_{1,1}\|^2}{N_0 + \bar{P}_2 \lambda_{\min}\LB\mathbf{R}_{1,2}\RB}\bigg)\RSB \geq \tau_1, \\
 & \;\;\;\;\;\;\;\;\;\;\;\;\;\;\;\;\;\;\;\;0 \leq \bar{P}_1 \leq P_1^{\max}, \hspace{0.1in} 0 \leq \bar{P}_2 \leq P_2^{\max}.
\end{aligned}
\tag{P3}
\label{ub_problem}
\end{equation}
The ergodic rate expressions appearing in \eqref{ub_problem} can be computed in closed form by applying Lemma~\ref{lemma_1}. The optimal slow power control values are obtained by exploiting Proposition~\ref{propos_2}. Indeed, one of the two TXs transmits with full power, while the other one controls its power by bisection. Comparing the performance and the feasibility of both solutions leads to the solution of optimization problem \eqref{ub_problem}.
\section{Numerical Evaluation}

With the aim of evaluating the performance of the proposed statistically coordinated precoding scheme, extensive Monte Carlo simulations have been performed and, more specifically, $20000$ channel realizations have been simulated. We assume the existence of $M_1=M_2=M=4$ antennas at each TX. Furthermore, we consider unit noise variance ($N_0=1$) and a QoS threshold~$\tau_1 = 1$~bps/Hz.

We consider a classical exponential channel correlation model \cite{Loyka2001}, in which the covariance matrices~$\mathbf{R}_{i,j}$ are given by
\begin{equation}
 \mathbf{R}_{i,j} = \beta_{i,j} \begin{bmatrix}
1&\rho&\rho^2&\ldots&\rho^{M-1}\\
\rho&1&\rho&\ldots&\rho^{M-2}\\
\vdots&\vdots&\vdots&\ddots&\vdots\\
\rho^{M-1}&\rho^{M-2}&\rho^{M-3}&\ldots&1\\
\end{bmatrix}, \hspace{0.1in} i,j \in \{1,2\},
\label{covar_struct}
\end{equation}
where, $\beta_{i,j}$ represents the pathloss and is chosen here equal to $1$ when $i=j$ and to $0.3$ otherwise. In the investigated scenario the antenna correlation factor, $\rho$, is equal to 0.5.

\begin{figure}[!ht]
  \centering
  \includegraphics[scale=0.70]{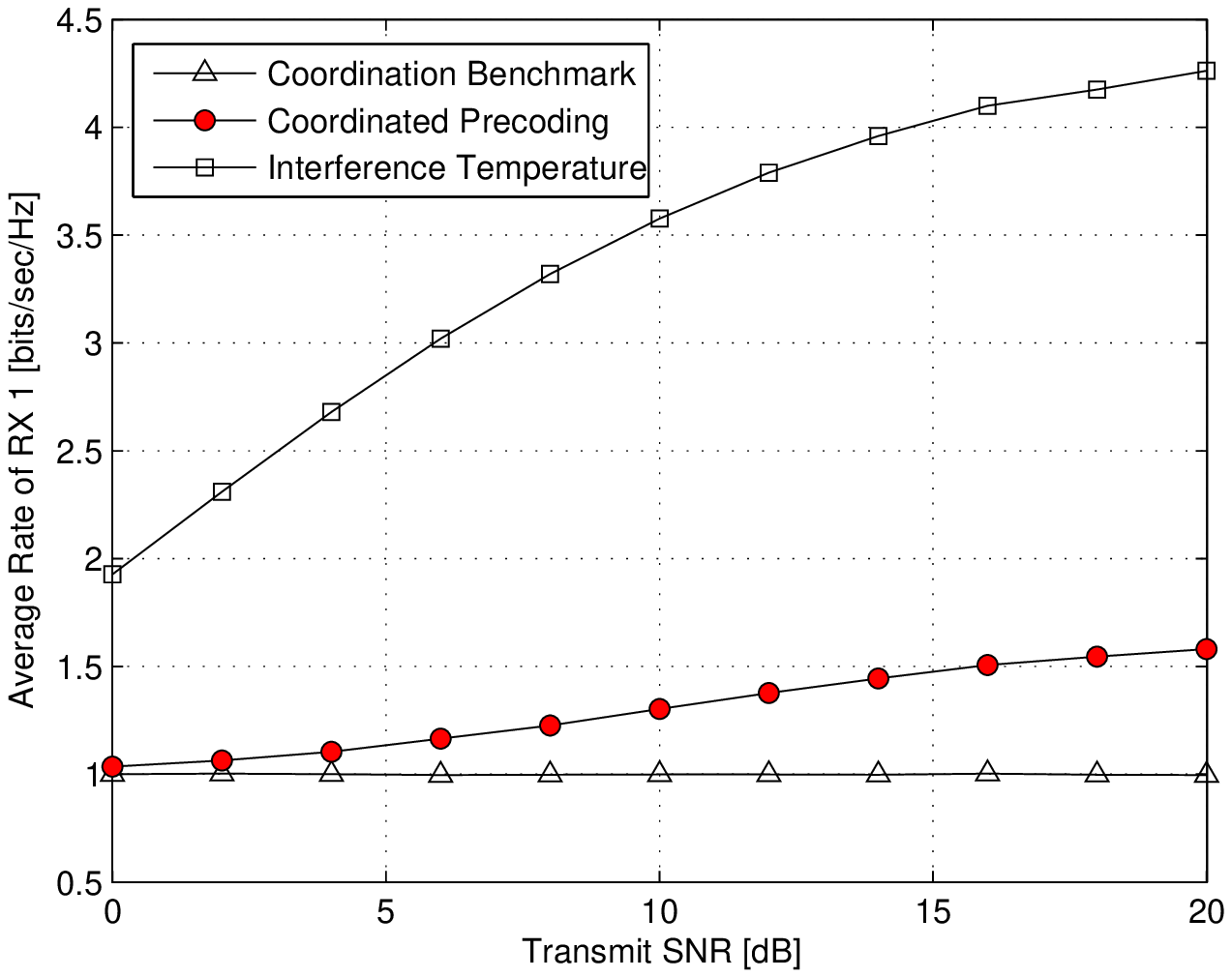}\\
	\vspace{-0.2in}
  \caption{Ergodic rate of ${\textrm{RX}}~1$ vs. transmit SNR, when incumbent QoS threshold $\tau_1=1$bps/Hz.} \label{fig:Erg_Rate_RX1_vs_SNR}
\end{figure}

\begin{figure}[!ht]
  \centering
  \includegraphics[scale=0.70]{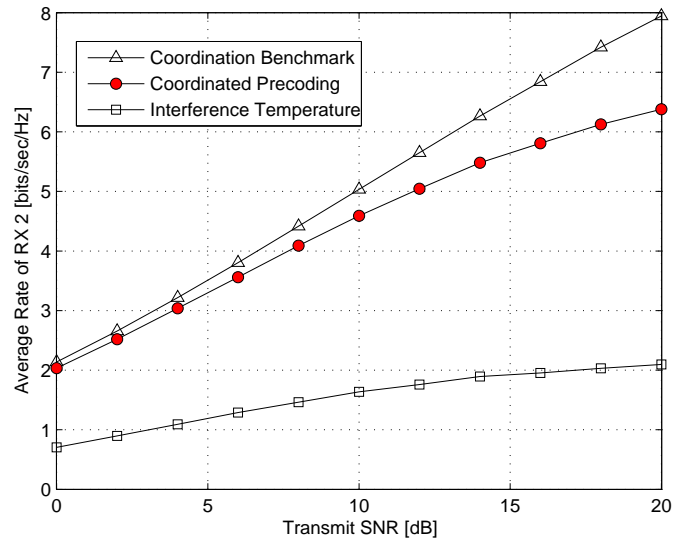}\\
	\vspace{-0.2in}
  \caption{Ergodic rate of ${\textrm{RX}}~2$ vs. transmit SNR, when incumbent QoS threshold $\tau_1=1$bps/Hz.}
\label{fig:Erg_Rate_RX2_vs_SNR}
\end{figure}

In Fig.~\ref{fig:Erg_Rate_RX1_vs_SNR} and Fig.~\ref{fig:Erg_Rate_RX2_vs_SNR}, the average rate of RX~$1$ and the average rate of ${\textrm{RX}}~2$ are depicted as a function of the system's transmit Signal-to-Noise Ratio (SNR). The three curves represent the throughput performance achieved by the proposed statistically coordinated precoding scheme, the interference temperature-based precoding scheme, as well as the described coordination benchmark. Focusing on RX~$2$, the coordination benchmark outperforms both the proposed precoding scheme as well as the interference temperature-based scheme, as expected. By observing Fig.~\ref{fig:Erg_Rate_RX1_vs_SNR}, it should be noted that, in contrast with the coordination benchmark, the proposed precoding scheme fails to satisfy the incumbent average rate constraint with equality. This occurs because we resort to tackling optimization problem \eqref{eq:prob_new}, which involves a lower bound of the average rate of RX~$1$. Nevertheless, the proposed algorithm successfully manages to control the average rate of RX~$1$ and this capability can be translated to a significant throughput gain for the licensee, in comparison to the one achieved by the interference temperature-based precoding scheme.

\begin{figure}[!ht]
  \centering
  \includegraphics[scale=0.70]{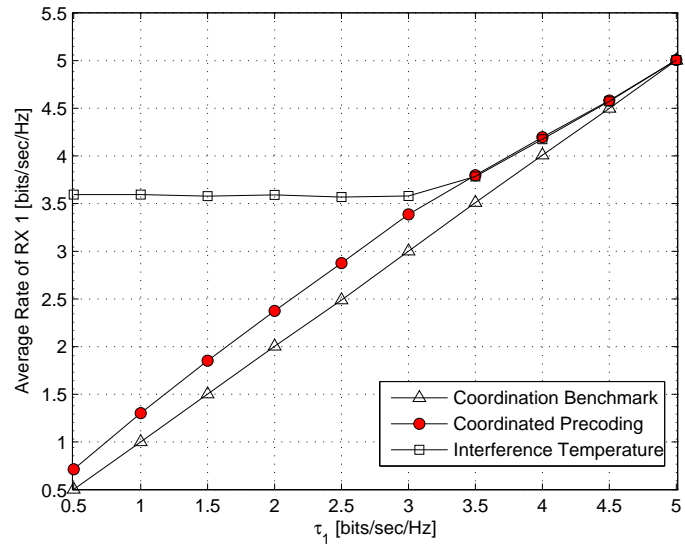}\\
	\vspace{-0.2in}
  \caption{Ergodic rate of ${\textrm{RX}}~1$ vs. threshold $\tau_1$, SNR=$10$dB.} 
\label{fig:Erg_Rate_RX1_vs_tau_1}
\end{figure}

\begin{figure}[!ht]
\centering
\includegraphics[scale=0.70]{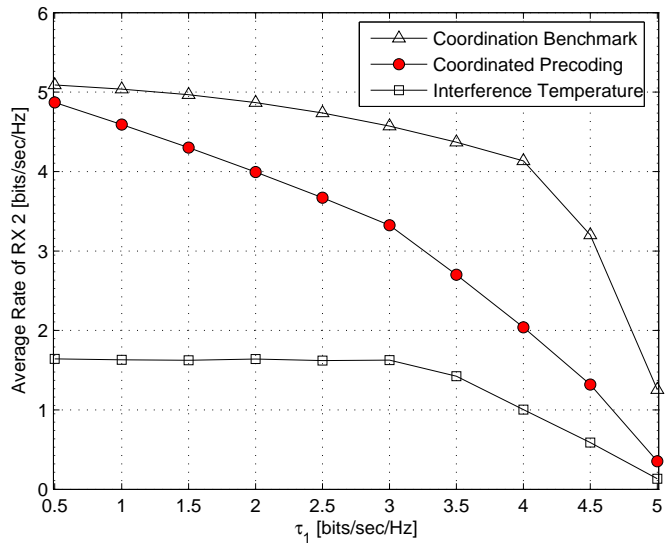}\\
\vspace{-0.2in}
\caption{Ergodic rate of ${\textrm{RX}}~2$ vs. threshold $\tau_1$, SNR=$10$dB.}
\label{fig:Erg_Rate_RX2_vs_tau_1}
\end{figure}

The achievable average rates of RX~$1$ and RX~$2$, by applying the proposed precoding algorithm, along with the ones achieved by the two reference precoding schemes, are depicted in Fig.~\ref{fig:Erg_Rate_RX1_vs_tau_1} and Fig.~\ref{fig:Erg_Rate_RX2_vs_tau_1}, respectively, as a function of QoS threshold, $\tau_1$, when the transmit SNR of the system is equal to $10$~dB. The average rate constraint for RX~$1$ is fulfilled by all three schemes for the whole examined range of $\tau_1$. Also, the proposed precoding scheme outperforms the interference temperature-based one, especially when the average rate constraint of the incumbent is loose, which occurs due to the fact that under this regime there is more to gain for the licensee by means of an efficient coordination.
\section{Conclusions}

In this work, we propose a novel, joint precoding scheme with reference to a shared spectrum access system where the two TXs coordinate on the basis of statistical knowledge of the global multi-user channel. Our approach consists in formulating a Team Decision problem, the solution of which is approached by reducing the transmission strategy space to a finite number of strategies. This method is key to enforcing between the TXs and obtaining a practical solution to the intricate Team Decision problem. Such an approach allows to improve over the conventional underlay CR approach, by enforcing more coordination between the two coexisting TXs at the price of low CSI and complexity requirements, as the coordination can be realized offline. Approaching the global optimum is both a difficult and challenging problem that can be tackled in the future. The proposed scheme has also a strong potential in other more complex scenarios with multiple incumbent and/or licensee networks.

\appendix
%
\begin{lemma}\cite[eq. (37)]{Mallik2004}\label{lemma_1}
Let $\bar{\gamma} \in \mathbb{R}^{+}$ and $\bm{h} \sim \mathcal{CN}(\bm{0}_n, \mathbf{R}_{\bm{h}})$, where covariance matrix $\mathbf{R}_{\bm{h}}$ has $n$ distinct eigenvalues ${\{\lambda_j\}}_{j=1}^n$. It then holds  
\begin{equation}
\E_{\bm{h}}\LSB \log_2(1 + \bar{\gamma} {\|\bm{h}\|}^2)\RSB = \frac{1}{\ln(2) \bar{\gamma} \prod_{j=1}^n \lambda_j} \sum_{j=1}^n \frac{\bar{\gamma} \lambda_j e^{\frac{1}{\bar{\gamma} \lambda_j}} E_1\big(\frac{1}{\bar{\gamma} \lambda_j}\big)}{\prod_{m=1, m \neq j}^n \big(\frac{1}{\lambda_m} - \frac{1}{\lambda_j}\big)}.
\end{equation}
\end{lemma} 

\begin{lemma}\cite[eq. (75)-(76)]{Wang2012}\label{lemma_2}
Let $\bar{\gamma} \in \mathbb{R}^{+}$ and $\bm{w} \in \mathbb{C}^{n \times 1}$ be deterministic, and $\bm{h} \sim \mathcal{CN}(\bm{0}_n, \mathbf{R}_{\bm{h}})$. It then holds
\begin{equation}
\mathbb{E}_{\bm{h}}\LSB\log_2\left(1 + \bar{\gamma}{|\bm{h}^{\He} \bm{w}|}^2\right)\RSB= \frac{1}{\ln(2)} e^{\frac{1}{\bar{\gamma} \lambda_1(\mathbf{R}_{{\textrm{eff}}})}} E_1\bigg(\frac{1}{\bar{\gamma} \lambda_1(\mathbf{R}_{{\textrm{eff}}})}\bigg),
\end{equation}
where $\lambda_1(\mathbf{R}_{{\textrm{eff}}})$ is the unique non-zero (positive) eigenvalue of matrix $\mathbf{R}_{{\textrm{eff}}} = \mathbf{R}_{\bm{h}}^{\frac{1}{2}} \bm{w} \bm{w}^{\He} \mathbf{R}_{\bm{h}}^{\frac{1}{2}}$.
\end{lemma}

\begin{lemma}\label{lemma_3}
Let us consider two positive semi-definite matrices $\mathbf{A}$ and $\mathbf{B}$ in $\mathbb{C}^{n\times n}$ with eigenvalues denoted as~$\lambda_1(\mathbf{A}),\ldots,\lambda_n(\mathbf{A})$ and~$\lambda_1(\mathbf{B}),\ldots,\lambda_n(\mathbf{B})$, respectively, where it is assumed that $\mathbf{A}$ is of full rank and has no multiple eigenvalues. We also assume that matrix $\mathbf{A}$ can be decomposed as $\mathbf{A} = \mathbf{U}_{\mathbf{A}} \mathbf{\Lambda}_{\mathbf{A}} \mathbf{U}_{\mathbf{A}}^{\He}$, where $\mathbf{U}_{\mathbf{A}}$ is a unitary matrix and $\mathbf{\Lambda}_{\mathbf{A}} = \diag\left(\lambda_1(\mathbf{A}),\ldots,\lambda_n(\mathbf{A})\right)$, where it holds that $0 < \lambda_1(\mathbf{A}) <\ldots< \lambda_n(\mathbf{A})$. Let $\bm{x} \in \mathbb{C}^{n \times n}$ be a standard complex Gaussian random vector, such that $\bm{x} \sim \mathcal{CN}(\mathbf{0}_n, \mathbf{I}_n)$. It then holds
\begin{equation}\label{closed_form_expec}
\begin{aligned}
 \mathbb{E}\LSB\frac{\bm{x}^{\He} \mathbf{B} \bm{x}}{\bm{x}^{\He} \mathbf{A} \bm{x}}\RSB &= \sum_{i=1}^n {[\tilde{\mathbf{B}}]}_{i,i} \left\{\frac{\lambda_i(\mathbf{A})^{n-2} \left((n-1) \left(\ln(\lambda_i(\mathbf{A})) - \gamma\right) + 1\right)}{\prod_{j \neq i} \left(\lambda_i(\mathbf{A}) - \lambda_j(\mathbf{A})\right)} \right. \\
&\qquad \left. {} - \frac{\lambda_i(\mathbf{A})^{n-1} \left(\ln(\lambda_i(\mathbf{A})) - \gamma\right) \sum_{r=1,r \neq i}^n \prod_{j \neq i,r} \left(\lambda_i(\mathbf{A}) - \lambda_j(\mathbf{A})\right)}{{\left(\prod_{j \neq i} \left(\lambda_i(\mathbf{A}) - \lambda_j(\mathbf{A})\right)\right)}^2} \right. \\ 
& \qquad \left. {} + \sum_{k=1,k \neq i}^n \frac{\lambda_k(\mathbf{A})^{n-1} \left(\ln(\lambda_k(\mathbf{A})) - \gamma\right) \prod_{j \neq k,i} \left(\lambda_k(\mathbf{A}) - \lambda_j(\mathbf{A})\right)}{{\left(\prod_{j \neq k} \left(\lambda_k(\mathbf{A}) - \lambda_j(\mathbf{A})\right)\right)}^2}\right\},
\end{aligned}
\end{equation}
where $\tilde{\mathbf{B}} = \mathbf{U}_{\mathbf{A}}^{\He} \mathbf{B} \mathbf{U}_{\mathbf{A}}$.
\begin{proof}
We prove this result in two steps. Firstly, we show that considering two matrices $\bA$ and $\bB$ with different eigenbases, we can come back to the case of matrices having the same eigenbasis. We then prove the lemma for this case.
 
Let us assume that $\mathbf{A}$ and $\mathbf{B}$ have \emph{different} eigenbases. We consider their eigendecompositions $\mathbf{A} = \mathbf{U}_{\mathbf{A}} \mathbf{\Lambda}_{\mathbf{A}} \mathbf{U}_{\mathbf{A}}^{\He}$ and $\mathbf{B} = \mathbf{U}_{\mathbf{B}} \mathbf{\Lambda}_{\mathbf{B}} \mathbf{U}_{\mathbf{B}}^{\He}$, where the diagonal entries of $\mathbf{\Lambda}_{\mathbf{A}}$ are sorted in an increasing order and the diagonal entries of $\mathbf{\Lambda}_{\mathbf{B}}$ are non-decreasingly sorted. Introducing matrix $\tilde{\mathbf{B}} = \mathbf{U}_{\mathbf{A}}^{\He} \mathbf{B} \mathbf{U}_{\mathbf{A}}$, the expectation in question becomes 
\begin{equation}
 \mathbb{E}\LSB\frac{\bm{x}^{\He} \mathbf{B} \bm{x}}{\bm{x}^{\He} \mathbf{A} \bm{x}}\RSB = \mathbb{E}\LSB\frac{\bm{x}^{\He} \mathbf{U}_{\mathbf{A}}^{\He} \mathbf{B} \mathbf{U}_{\mathbf{A}} \bm{x}}{\bm{x}^{\He} \mathbf{\Lambda}_{\mathbf{A}} \bm{x}}\RSB = \mathbb{E}\LSB\frac{\bm{x}^{\He} \tilde{\mathbf{B}} \bm{x}}{\bm{x}^{\He} \mathbf{\Lambda}_{\mathbf{A}} \bm{x}}\RSB = \mathbb{E}\LSB\frac{\sum_{i=1}^n \sum_{j=1}^n {[\tilde{B}]}_{i,j} x_i^{*} x_j}{\sum_{k=1}^n \lambda_k(\mathbf{A}) |x_k|^2}\RSB.
\label{diff_eigenbase_1}
\end{equation}
Exploiting the fact that $\bm{x} \sim \mathcal{CN}(\bm{0}_n, \mathbf{I}_n)$, if we write each $x_i, \hspace{0.1in} i=1,\ldots,n$ in polar representation, i.e., $x_i = |x_i| e^{j \phi_i}$, then we have that all phases $\phi_i, \hspace{0.1in} i=1,\ldots,n$ and amplitudes $|x_i|, \hspace{0.1in} i=1,\ldots,n$ are mutually independent and the phases are uniformly distributed. As a result, the expectation takes the following form
\begin{equation}
 \mathbb{E}\LSB\frac{\bm{x}^{\He} \mathbf{B} \bm{x}}{\bm{x}^{\He} \mathbf{A} \bm{x}}\RSB = \mathbb{E}\LSB\frac{\sum_{i=1}^n \sum_{j=1}^n {[\tilde{B}]}_{i,j} |x_i| |x_j| e^{j(\phi_j - \phi_i)}}{\sum_{k=1}^n \lambda_k(\mathbf{A}) |x_k|^2}\RSB = \mathbb{E}\LSB\frac{\sum_{i=1}^n {[\tilde{B}]}_{i,i} |x_i|^2}{\sum_{k=1}^n \lambda_k(\mathbf{A}) |x_k|^2}\RSB,
\label{diff_eigenbase_2}
\end{equation}
or, equivalently
\begin{equation}
 \mathbb{E}\LSB\frac{\bm{x}^{\He} \mathbf{B} \bm{x}}{\bm{x}^{\He} \mathbf{A} \bm{x}}\RSB = \mathbb{E}\LSB\frac{\bm{x}^{\He} \mathbf{U}_{\mathbf{A}} \diag\left({[\tilde{B}]}_{1,1},\ldots,{[\tilde{B}]}_{n,n}\right) \mathbf{U}_{\mathbf{A}}^{\He} \bm{x}}{\bm{x}^{\He} \mathbf{A} \bm{x}}\RSB.
\label{diff_eigenbase_3}
\end{equation}
Hence, the case of \emph{equal} eigenbases is recovered.

Consequently, we proceed by considering without loss of generality that matrices $\mathbf{A}$ and $\mathbf{B}$ have the same eigenbases. However, it should be noted that elements ${[\tilde{B}]}_{i,i}, \hspace{0.1in} i=1,\ldots,n$ are not sorted in any particular order.

Focusing, now, on the derivation of a closed form expression of the expectation, we have that
\begin{equation}
\begin{aligned}
 \mathbb{E}\LSB\frac{\bm{x}^{\He} \mathbf{B} \bm{x}}{\bm{x}^{\He} \mathbf{A} \bm{x}}\RSB &= \mathbb{E}\LSB\frac{\sum_{i=1}^n {[\tilde{B}]}_{i,i} |x_i|^2}{\sum_{j=1}^n \lambda_j(\mathbf{A}) |x_j|^2}\RSB = \sum_{i=1}^n {[\tilde{B}]}_{i,i} \mathbb{E}\LSB\frac{|x_i|^2}{\sum_{j=1}^n \lambda_j(\mathbf{A}) |x_j|^2}\RSB \\
&= \sum_{i=1}^n {[\tilde{B}]}_{i,i} \mathbb{E}\LSB\frac{\partial}{\partial \lambda_i(\mathbf{A})} \ln\left(\sum_{j=1}^n \lambda_j(\mathbf{A}) |x_j|^2 \right)\RSB \\
&= \sum_{i=1}^n {[\tilde{B}]}_{i,i} \frac{\partial}{\partial \lambda_i(\mathbf{A})} \mathbb{E}\LSB\ln\left(\sum_{j=1}^n \lambda_j(\mathbf{A}) |x_j|^2 \right)\RSB.
\end{aligned}
\label{expec_1}
\end{equation}
Let us define random variable (RV) $X \triangleq \sum_{j=1}^n \lambda_j(\mathbf{A}) |x_j|^2$.~Using \cite[eq. (8)]{Abbe2012}, it is shown by induction that the Probability Density Function (PDF) of $X$ is the following
\begin{equation}
 p_X(x) = \sum_{k=1}^n \prod_{j \neq k} \frac{\lambda_k(\mathbf{A})^{n-2}}{\lambda_k(\mathbf{A}) - \lambda_j(\mathbf{A})} e^{-\frac{x}{\lambda_k(\mathbf{A})}}.
\label{pdf_X}
\end{equation}
As a result, the expectation of RV $\ln(X)$ is given by the expression that follows
\begin{equation}
 \mathbb{E}\LSB\ln(X)\RSB = \sum_{k=1}^n \int_0^{\infty} \ln(x) e^{-\frac{x}{\lambda_k(\mathbf{A})}} dx \prod_{j \neq k} \frac{\lambda_k(\mathbf{A})^{n-2}}{\lambda_k(\mathbf{A}) - \lambda_j(\mathbf{A})}.
\label{expect_2}
\end{equation}
Exploiting \cite[eq. (4.331.1)]{Gradshteyn2007}, expression \eqref{expect_2} becomes
\begin{equation}
 \mathbb{E}\LSB\ln(X)\RSB = \sum_{k=1}^n \frac{\lambda_k(\mathbf{A})^{n-1} \left(\ln(\lambda_k(\mathbf{A})) - \gamma\right)}{\prod_{j \neq k} \left(\lambda_k(\mathbf{A}) - \lambda_j(\mathbf{A})\right)}.
\label{expect_3}
\end{equation}
Taking, now, the partial derivative of \eqref{expect_3}, with respect to $\lambda_i(\mathbf{A})$, we obtain
\begin{equation}
\begin{aligned}
 \frac{\partial}{\partial \lambda_i(\mathbf{A})} \mathbb{E}\LSB\ln(X)\RSB &= \frac{\partial}{\partial \lambda_i(\mathbf{A})}\left\{\frac{\lambda_i(\mathbf{A})^{n-1} \left(\ln(\lambda_i(\mathbf{A})) - \gamma\right)}{\prod_{j \neq i} \left(\lambda_i(\mathbf{A}) - \lambda_j(\mathbf{A})\right)}\right\} \\
&+ \frac{\partial}{\partial \lambda_i(\mathbf{A})}\left\{\sum_{k=1,k \neq i}^n \frac{\lambda_k(\mathbf{A})^{n-1} \left(\ln(\lambda_k(\mathbf{A})) - \gamma\right)}{\prod_{j \neq k} \left(\lambda_k(\mathbf{A}) - \lambda_j(\mathbf{A})\right)}\right\}.
\end{aligned}
\label{partial_der_1}
\end{equation}
For the first term of \eqref{partial_der_1}, the following expression is obtained
\begin{equation}
\begin{aligned}
 \frac{\partial}{\partial \lambda_i(\mathbf{A})}\left\{\frac{\lambda_i(\mathbf{A})^{n-1} \left(\ln(\lambda_i(\mathbf{A})) - \gamma\right)}{\prod_{j \neq i} \left(\lambda_i(\mathbf{A}) - \lambda_j(\mathbf{A})\right)}\right\} &= \frac{\lambda_i(\mathbf{A})^{n-2} \left((n-1) \left(\ln(\lambda_i(\mathbf{A})) - \gamma\right) + 1\right)}{\prod_{j \neq i} \left(\lambda_i(\mathbf{A}) - \lambda_j(\mathbf{A})\right)} \\
&- \frac{\lambda_i(\mathbf{A})^{n-1} \left(\ln(\lambda_i(\mathbf{A})) - \gamma\right) \sum_{r=1,r \neq i}^n \prod_{j \neq i,r} \left(\lambda_i(\mathbf{A}) - \lambda_j(\mathbf{A})\right)}{{\left(\prod_{j \neq i} \left(\lambda_i(\mathbf{A}) - \lambda_j(\mathbf{A})\right)\right)}^2}.
\end{aligned}
\label{partial_der_2}
\end{equation}
It now remains to find the second term of \eqref{partial_der_1} in closed form.~We, thus, obtain the following
\begin{equation}
 \frac{\partial}{\partial \lambda_i(\mathbf{A})}\left\{\sum_{k=1,k \neq i}^n \frac{\lambda_k(\mathbf{A})^{n-1} \left(\ln(\lambda_k(\mathbf{A})) - \gamma\right)}{\prod_{j \neq k} \left(\lambda_k(\mathbf{A}) - \lambda_j(\mathbf{A})\right)}\right\} = \sum_{k=1, k \neq i}^n \frac{\partial}{\partial \lambda_i(\mathbf{A})} \left\{\frac{\lambda_k(\mathbf{A})^{n-1} \left(\ln(\lambda_k(\mathbf{A})) - \gamma\right)}{\prod_{j \neq k}\left(\lambda_k(\mathbf{A}) - \lambda_j(\mathbf{A})\right) }\right\}.
\label{partial_der_3}
\end{equation}
Given that $i \neq k$, the partial derivative appearing in the right hand side of \eqref{partial_der_3}, is given by the following expression
\small
\begin{equation}
 \frac{\partial}{\partial \lambda_i(\mathbf{A})} \left\{\frac{\lambda_k(\mathbf{A})^{n-1} \left(\ln(\lambda_k(\mathbf{A})) - \gamma\right)}{\prod_{j \neq k}\left(\lambda_k(\mathbf{A}) - \lambda_j(\mathbf{A})\right) }\right\} = - \frac{\lambda_k(\mathbf{A})^{n-1} \left(\ln(\lambda_k(\mathbf{A})) - \gamma\right) \frac{\partial}{\partial \lambda_i(\mathbf{A})} \left\{\prod_{j \neq k}\left(\lambda_k(\mathbf{A}) - \lambda_j(\mathbf{A}) \right) \right\} }{{\left(\prod_{j \neq k} \left(\lambda_k(\mathbf{A}) - \lambda_j(\mathbf{A})\right)\right)}^2},
\label{partial_der_4}
\end{equation}
\normalsize
where $\frac{\partial}{\partial \lambda_i(\mathbf{A})} \left\{\prod_{j \neq k}\left(\lambda_k(\mathbf{A}) - \lambda_j(\mathbf{A}) \right) \right\} = - \prod_{j \neq k,i} \left(\lambda_k(\mathbf{A}) - \lambda_j(\mathbf{A})\right)$.~This concludes the proof.
\end{proof} 
\end{lemma}

\ifCLASSOPTIONcaptionsoff
  \newpage
\fi

\bibliography{bib/allCitations}
\bibliographystyle{IEEEtran}

\end{document}

%% file: Definitions.tex
\DeclareMathAlphabet{\mathbit}{OML}{cmr}{bx}{it}
\DeclareMathAlphabet{\mathsf}{OT1}{cmss}{m}{n}
\DeclareMathAlphabet{\mathTXf}{OT1}{cmss}{bx}{it}

\DeclareMathOperator{\diag}{diag}

\newcommand{\bA}{\mathbf{A}} 
\newcommand{\bB}{\mathbf{B}}

\newcommand{\bR}{\mathbf{R}}

\newcommand{\bh}{\bm{h}}

\newcommand{\LB}{\left(}
\newcommand{\RB}{\right)}
\newcommand{\LSB}{\left[}
\newcommand{\RSB}{\right]}

\newcommand{\tr}{{\text{tr}}}

\newcommand{\E}{{\mathbb{E}}}

\newcommand{\He}{{{\mathrm{H}}}}

\theoremstyle{remark}
\newtheorem{remark}{Remark} 